\documentclass[copyright,creativecommons]{eptcs}
\providecommand{\event}{WWV 2013} 
\usepackage{breakurl}             
\usepackage{epsfig}
\usepackage{amsfonts}
\usepackage{amssymb,amsthm}
\usepackage{color}
\usepackage[leqno]{amsmath,stmaryrd}
\usepackage{txfonts}
\usepackage{verbatim}
\usepackage{enumerate}
\usepackage{hyperref}
\usepackage{prooftree}
\usepackage{cmll}
\usepackage{syndmeta}
\usepackage{uris}
\usepackage{xypic}
\usepackage{fancybox}
\usepackage{wrapfig}

\newtheorem{theorem}{Theorem}[section]

\begin{document}

\title{
 Local Type Checking for Linked Data Consumers
}

\author{
 Gabriel Ciobanu \quad Ross Horne
\institute{
 Romanian Academy, Institute of Computer Science, 
 Blvd. Carol I, no. 8, 700505 Ia{\c{s}}i, Romania
}
 \email{gabriel@info.uaic.ro \quad ross.horne@gmail.com}
\and
 Vladimiro Sassone
\institute{
 University of Southampton, Electronics and Computer Science, 
 Southampton, United Kingdom
}
 \email{vs@ecs.soton.ac.uk}
}

\def\titlerunning{
 Local Type Checking for Linked Data Consumers
}
\def\authorrunning{G. Ciobanu, R. Horne \& V. Sassone}

\maketitle

\begin{abstract}
The Web of Linked Data is the cumulation of over a decade of work by the Web standards community in their effort to make data more Web-like. We provide an introduction to the Web of Linked Data from the perspective of a Web developer that would like to build an application using Linked Data. We identify a weakness in the development stack as being a lack of domain specific scripting languages for designing background processes that consume Linked Data. To address this weakness, we design a scripting language with a simple but appropriate type system. In our proposed architecture some data is consumed from sources outside of the control of the system and some data is held locally. Stronger type assumptions can be made about the local data than external data, hence our type system mixes static and dynamic typing. Throughout, we relate our work to the W3C recommendations that drive Linked Data, so our syntax is accessible to Web developers.

\end{abstract}

\section{Introduction}

Linked data is raw data published on the Web that makes use of URIs to establish links between datasets. The use of URIs to identify resources allows data about resources to be looked up (dereferenced) using a simple protocol, and for the data returned to contain more URIs that can also be looked up. Linked Data consumers can crawl the Web of Linked Data to pull in data that can enrich a Web application. For example the 2012 Olympics Website used Linked Data to help journalists discover and organise statistics about relatively unknown medal winners during the games. Due to links bringing down barriers between datasets, the Web of Linked Data is amongst the worlds largest datasets in the hands of Web developers.

We describe a simple, but appropriate architecture for a Linked Data consumer. We then address the problem of designing a type system for this architecture. Linked Data published on the Web from multiple sources is inherently messy, so data arriving over HTTP must be dynamically type checked. Dynamically type checked data is then loaded into a local triple store. The local triple store persists a view of the Web of Linked Data relevant to the Linked Data consumer. Once the consumed Linked Data has been dynamically typed checked, queries and scripts over the local data can be type checked using a mix of dynamic and static type checking. Finally, because we can never have a global view of the Web of Linked Data, we take care to design our type system so that type checking is local. To achieve this we select a useful subset of the W3C standards SPARQL, RDF Schema and OWL to design our type system.

In Section~\ref{section:background}, we describe the Linked Data architecture, with an emphasis on consuming Linked Data. In Section~\ref{section:types} we argue for a notion of type that aligns with the relevant W3C recommendations and is as simple as possible whilst picking up basic programming errors. In Section~\ref{section:system}, we define the syntax of our scripting language for consuming Linked Data and the rules of our type system.


\section{From a Web of Documents to a Web of Data}
\label{section:background}

In 1989, Tim Berners-Lee proposed a hypertext system that became the World Wide Web.
In his proposal~\cite{Berners-Lee1989}, he 
observes that hypertext systems from the 1980's failed to gain traction 
because they attempted to justify themselves on their own data. He 
described a simple but effective architecture for exposing existing data from 
file systems and databases as HTML documents that link to 
other documents using URIs. This architecture is still used today to 
present documents that link to documents. 

Despite the success of the World Wide Web, Berners-Lee was not completely 
satisfied. He also wanted to make raw data itself Web-like, not just the 
documents that present data. His first vision was called the 
Semantic Web~\cite{Berners-Lee2001}, which was an AI textbook vision of a world 
where intelligent agents would understand data on the Web to do every day 
tasks on our behalf. As admitted by Berners-Lee and his co-author Hendler, 
there was much hype but limited success scaling ideas. 
Hendler self-critically asked: ``Where are all 
the intelligent agents?"~\cite{Hendler2007}. By 2006~\cite{Shadbolt2006}, Berners-Lee
had come to the conclusion that there had been too much 
emphasis on deep ontologies and not enough data.

Thus Berners-Lee returned to the grass roots of the Web: the Web developers. 
He described a simple protocol for publishing raw data on the 
Web~\cite{Berners-Lee2006}. The protocol makes use of standards, namely URIs 
as global identifiers, HTTP as a transport layer and RDF as a data format, according to the following principles:
\begin{itemize}
\item use URI to identify resources (i.e.\ anything that might be referred to in data),
\item use HTTP URIs to identify resources so we can look them up (using the HTTP GET verb),
\item when a URI is looked up, return data about the resource using the standards (RDF),
\item include URIs in the data, so they can also be looked up.
\end{itemize}
Data published according to the above protocol is called \textit{Linked 
Data}. An HTTP URI that returns data when it is looked up is a 
\textit{dereferenceable} URI. All URIs that appear in this paper are 
dereferenceable, so are part of this rapidly growing Web of Linked 
Data~\cite{Heath2011}.

The Linked Data protocol is one example of a \textit{RESTful} 
protocol~\cite{Fielding2002}. A RESTful protocol runs directly on the HTTP 
protocol using the HTTP verbs (including GET, PUT and DELETE) which are 
suited to services of publishing data. Many data protocols such as the 
Twitter API\footnote{Twitter API: \url{https://dev.twitter.com/docs/api/1.1}}, Facebook Open Graph protocol\footnote{Facebook Open Graph protocol: \url{https://developers.facebook.com/docs/opengraph/}} and the Google Data API\footnote{Google Data API: \url{https://developers.google.com/gdata/}} are RESTful, and 
with some creativity they can be broadly interpreted as Linked Data. Hence, like the Web of 
hypertext, the Web of Linked Data does not need to justify itself solely on 
its own data.

\subsection{An Architecture for Consuming Linked Data}

Data owners may want to publish their data as Linked Data. Publishing 
Linked Data is no more difficult than building a traditional Web page. The 
developer should provide an RDF view of a dataset rather than an HTML 
view~\cite{Bizer2004}. Data from diverse sources such as 
Wikipedia~\cite{Bizer2009c} and governments~\cite{Shadbolt2012} 
can be lifted to the Web of Linked Data.

Data consumers may not own data, but have a data centric service to deliver.
 Data consumers can consume data from many Linked Data sources then 
exploit links between datasets. Consuming Linked Data is the main focus in 
our work. In Figure~\ref{figure:architecture}, we describe a simple 
architecture for an application that consumes Linked Data. The architecture 
is an extension of the traditional Web architecture.

\begin{figure}[hbt]
\quad
\xymatrix{
*++[F-,]{
 \mbox{RDF}
}
&&
\txt{
  background \\
  processes
}
\ar[ll]_-{\mbox{REST}}
\ar[rr]^-{\mbox{SPARQL}}
&&
*++[F-:<3pt>]{
 \mbox{triple store}
}
&&
\txt{
  front end
}
\ar[ll]_-{\mbox{SPARQL}}
\ar[r]^-{\mbox{Ajax}~~}
&
*++[F-,]{\mbox{HTML}}
}
\caption{A simple but effective architecture for an application that consumes Linked Data.
}
\label{figure:architecture}
\end{figure}
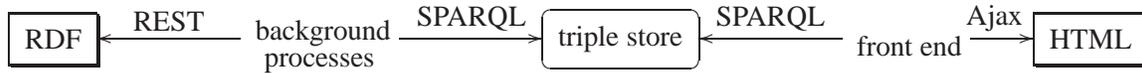
At the heart of our application is a \textit{triple store}, which replaces the more 
traditional relational database. A triple store is optimized for storing 
RDF data in subject-property-object form corresponding to a labelled edge 
in a graph. There are several commercial grade triple stores including 
Sesame, Virtuoso and 4store~\cite{Broekstra2002,Erling2007,Harris2009}, which can operate 
on scales of billions of triples, i.e.\ enough for almost any Web 
application.

The front end of our application follows the traditional Web architecture 
pattern, where Web pages are generated from a database using scripts. The 
only difference is that our application uses SPARQL Query instead of SQL to read 
from the triple store. The syntax of SPARQL is similar to SQL, hence it is 
easy for an experienced Web developer to develop the front end. Most 
popular Web development frameworks, such as Ruby on Rails, have been 
extended to support SPARQL. The main reason that SQL is replaced with 
SPARQL is that a triple store typically stores data from heterogeneous data 
sources. Heterogeneous data sources are difficult to combine and query 
using tables with relational schema; whereas combining and querying graph 
models is more straightforward. Thus links between datasets can be more easily exploited in queries.


The key novel feature of an application that consumes Linked Data is the 
back end. At the back end, background processes can crawl the Web of Linked 
Data to discover new and relevant Linked Data sources. Back end processes also keep local Linked 
Data up-to-date. For example, data from a news feed may change hourly and changes are made to Wikipedia several times every second.
To consume data relevant to the application, the background processes should 
be programmable. This work focuses in particular on high level programming 
languages that can be used when programming the back end of an application. The 
background processes must be able to dereference Linked Data and update 
data in the triple store, as well as make decisions based on query results.


\subsection{A Low Level Approach to Consuming Linked Data}

The back end of the application that consumes Linked Data should keep track of every URI  
accessed through the HTTP protocol. The HTTP header of an HTTP response 
contains information that can be used for discovering and maintaining 
Linked Data about a given URI. We describe the dereferencing of URIs at a 
low level, to be concrete about what a high level language should abstract.

Consider an illustrative example of dereferencing a URI. If we perform an 
HTTP GET on the URI \xres{Kazakhstan} with an HTTP header indicating that 
we accept the mime type \texttt{text/n3}, then we get a {\it 303 See Other 
response}.\footnote{Reproducible using: \texttt{curl -v -I -H 
"Accept:text/n3" http://dbpedia.org/resource/Kazakhstan}} A {\it 303 See Other 
response} means that you can get data of the serialisation type you 
requested at another location. If we now perform an HTTP GET request at the 
URI indicated by the {\it 303 See Other response} 
(\url{http://dbpedia.org/data/Kazakhstan.n3}), we get an HTTP 200 OK 
response including the following headers:
\begin{gather*}
\begin{array}{l}
\texttt{GET /data/Kazakhstan.n3 HTTP/1.1} \\
\texttt{Host: dbpedia.org} \\
\texttt{Accept: text/n3}
\end{array}
\qquad
\begin{array}{l}
\texttt{HTTP/1.1 200 OK} \\
\texttt{Date: Tue, 26 Mar 2013 15:39:49 GMT} \\
\texttt{Content-Type: text/n3}
\end{array}
\end{gather*}

From the response header (above right) we can tell that this URI 
successfully returned some RDF using the \texttt{n3} serialisation format. 
However, although the data was obtained from the second URI, the resource 
represented by \!\xres{Kazakhstan}\! is described in the data. 

The {\it 303 See Other response} is one of several ways Linked Data may be 
published using a RESTful protocol~\cite{Heath2011}. Furthermore, some 
systems such as the Virtuoso triple store~\cite{Erling2007} use wrappers to 
extract RDF data from other data sources, such as the Google Data API or 
the Twitter API. 
The programmer of a script that consumes Linked Data should not worry about 
details such as wrappers or the serialisation formats. Unfortunately, 
existing libraries for popular programming 
languages~\cite{Beckett2002,McBride2002} are at a low level of 
abstraction. We propose that a higher level language can hide the above 
details in a compiler that tries automatically to dereference URIs.

\subsection{A High Level Approach to Consuming Linked Data}

Here we consider languages that consume Linked Data at a higher level of 
abstraction. The only essential information is the URI to dereference and the data returned. 
Other features of a high level language include control flow and queries to decide what other URIs to dereference.

Consider the following script (based loosely on SPARQL~\cite{Harris2013}).
The keyword \texttt{select} binds a term bound to variable $\var{x}$.
The \texttt{from named} keyword indicates that we dereference the URI indicated.
The \texttt{where} keyword indicates that we would like to match a given triple pattern.
Notice that the second \texttt{from named} keyword dereferences a URI bound to $\var{x}$ that is not known until the query patter is matched. 
\[
\script{
 \fromnamed{\xres{Kazakhstan}} \\
 \select{\var{x}} \\
 \where{
  \nquad{\xres{Kazakhstan}}{\xres{Kazakhstan}}{\xdbp{capital}}{\var{x}}
 } \\
 \fromnamed{\var{x}} 
}
\]
The above script first ensures that data representing \!\xres{Kazakhstan}\! is stored in the named graph, also named \!\xres{Kazakhstan}\!, 
in the local triple store. If the URI has not been accessed before, then the 
URI is dereferenced. If the URI is dereferenced successfully, then the RDF 
data is stored in a named graph in the local triple store. 
Thus the data can be queried directly from the named graph in the local triple store.
This gives our applications a local view of the Web of Linked Data. Note 
that, if the URI is not dereferenced successfully, then this information can 
be recorded to avoid future attempts to dereference the URI.

Having dereferenced the first URI, the query indicated in the \texttt{where} clause is evaluated.
The query consists of a named graph to read, indicated by the \texttt{graph} keyword, and a triple pattern.
In the pattern, the subject is the resource \!\xres{Kazakhstan}\!, the property is \!\xdbp{capital}\!, and the object is not bound.
This query should find 
exactly one binding for $\var{x}$ to proceed. The query proceeds by 
discovering a URI (in this case \!\xres{Astana}\!), and substituting this URI for 
$\var{x}$ everywhere in the script. After the substitution, the second \texttt{from named} keyword
dereferences the discovered URI and loads the data into the triple store, as described above.

The above script abstracts away several HTTP GET requests and possibly some 
redirects and mappings between formats. It also encapsulates details where 
the following SPARQL Query is sent to the local SPARQL endpoint.
\[
\begin{array}{l}
\select{\var{x}}~\fromnamed{\xres{Kazakhstan}}~\texttt{where}
\\\qquad\{~\nquad{\xres{Kazakhstan}}{\xres{Kazakhstan}}{\xdbp{capital}}{\var{x}}~\}~\texttt{limit}~1
\end{array}
\]
At a lower level of abstraction, the above query would return a results document 
indicating that $\var{x}$ has one possible binding. Another programming language 
would extract the binding from the results document, then use the binding in some code that dereferences the URI.
Just doing this simple task in Java for example takes 
several pages of code. The Java program also involves treating parts of the 
code, such as the above query as a raw string of characters, which means 
that even basic errors parsing the syntax cannot be checked at compile time.

We argue that using the high level script presented at the beginning of 
this section is simpler than using a general purpose language
with libraries concerned with details of HTTP GET requests, constructing queries from 
strings and extracting variable bindings from query results.
Furthermore, it is worth noting that the syntax of the script does not significantly 
depart from the syntax of the SPARQL recommendations. In this way, we 
explore the idea of a domain specific scripting language for background 
processes of an application that consumes Linked Data.

\section{Simple Types in W3C Recommendations}
\label{section:types}

The W3C recommendations do not explicitly introduce a type system for 
Linked Data. However, there are some ideas in the RDF Schema~\cite{Brickley2004} and 
OWL~\cite{Hitzler2012} recommendations that can be used as a basis of a 
type system. 
Here we identify one of the simplest notion of a type, and justify the 
choice with respect to recommendations. The chosen notion of a type fits 
with types in Facebook's Open Graph.

\paragraph{Simple Datatypes.}
The only common component of the W3C recommendations in this section is the 
notion of a simple datatype. Simple datatypes are specified as part of the 
XML Schema recommendation~\cite{Biron2004}. A type system for Linked Data 
should be aware of the simple datatypes that most commonly appear, in 
particular \!\xsd{string}\!, \!\xsd{integer}\!, \!\xsd{decimal} and \!\xsd{dateTime}\!. 
All these types draw from disjoint lexical spaces, except \xsd{integer}\! is a 
subtype of \xsd{decimal}\!.
Note that we assume that \!\xsd{decimal}\!, \!\xsd{float}\! and \!\xsd{double}\! are 
different lexical representations of an abstract numeric type.

In the XML Schema recommendation, there is a simple datatype \!\xsd{anyURI}\!. 
This type is rarely actively used in ontologies -- the ontology for DBpedia 
only uses this type once as the range of the property \!\xdbp{review}\!. 
However, it unambiguously refers to any URI and nothing else, unlike 
\!\xrdfs{Resource}\! and \!\xowl{Thing}\! which, depending on the interpretation, may 
refer to more than just URIs or a subset of URIs. In this way, our type 
system is based on unambiguous simple datatypes, that frequently appear in 
datasets, such as DBpedia~\cite{Bizer2009c}.

\paragraph{Resource Description Framework.}
A URI is successfully dereferenced when we get a document from which we can 
extract RDF data~\cite{Carroll2004}. The basic unit of RDF is the 
\textit{triple}. An RDF triple consists of a subject, a property and an 
object. 
The subject, property and object may be URIs, and the object may also be a 
simple datatype. Most triple stores support \textit{quadruples}, where a 
fourth URI represents either the \textit{named graph}~\cite{Carroll2005} or 
the \textit{context} from where the triple was obtained.


Note that the RDF recommendation allows nodes with a local identifier, 
called a blank node, in place of a URI. Blank nodes are frequently debated 
in the community~\cite{Mallea2011}, due to several problems they introduce. 
Firstly, deciding the equality of graphs with blank nodes is an NP-complete 
problem; and, more seriously, when data with blank nodes is consumed more 
than once, each time the blank node is treated as a new blank node. This can 
cause many unnecessary copies of the same data to be created. 
We assume that our system assigns a new URI to each blank 
node in data consumed, hence do not introduce blank nodes into the local 
data model.

It is also important to note that the RDF specification has a vocabulary of 
URIs that have a distinguished role. Notably, the property \!\type\! is 
used to indicate a URI that classifies the resource. For example, the 
triple \!{\Kaz}{\type}{\xdbp{Country}}\! classifies the resource \!\Kaz\! as a 
\!\xdbp{Country}\!. Although the word ``type" is part of the URI for the 
property, we consider this triple to be part of the data format rather than 
part of our type system. In RDF, the word ``type" is used in the AI sense 
of a semantic network~\cite{Sowa2000}, rather than in a type theoretic 
sense. Since these ``types" can be changed like any other data, we make the 
design decision not to include them in our type system, because a 
type system is used for static analysis.

\paragraph{RDF Schema.}

The RDF Schema recommendation~\cite{Brickley2004} provides a core vocabulary for classifying resources using RDF. From this vocabulary we borrow only the top level class \!\xrdfs{Resource}\! and
the property \!\xrdfs{range}\!. All URIs are considered to identify resources, hence we equate \!\xrdfs{Resource}\! and \xsd{anyURI}\!. 
We define property types for URIs that are used in the property position of a triple.
A property type restricts the type of term that can be used in the object 
position of a triple. For example, according to the DBpedia ontology, the 
property \!\xdbp{populationDensity}\! has a range \!\xsd{decimal}\!. Thus our type 
system should accept that \!{\Kaz}{\xdbp{populationDensity}}\! $5.94$ is well 
typed. However, the type system should reject a triple with object 
\mstring{5.94} which is a string. 
We use the notation $\Property{\xsd{decimal}}$ for URIs representing properties permitting numbers as objects.

The RDF Schema \!\xrdfs{domain}\! of a property is redundant for our type 
system because only URIs can appear as the subject of a triple, and all URIs 
are of type \!\Res\!.
Note that properties are resources because they may appear in data. For 
example, the triple \xdbp{populationDensity} \xrdfs{label} 
\literal{population density (/sqkm)}{en} provides a description 
of the property in English.

\paragraph{Web Ontology Language.}

The Web Ontology Language (OWL)~\cite{Hitzler2012} is mostly concerned 
with classifying resources, which is not part of our type system. The OWL 
classes that are related to our type system are \!\xowl{ObjectProperty}\!, \!\xowl{DataTypeProperty}\! and \!\xowl{Thing}\!. An \!\xowl{ObjectProperty}\! is a 
property permitting URIs as objects, i.e.\ the type 
$\Property{\Res}$ in our type system. An \!\xowl{DataTypeProperty}\! is a 
property with one of the simple datatypes as its value.

In OWL~\cite{Hitzler2012}, \!\xowl{Thing}\! represents resources that are neither 
properties nor classes. We decide to equate \!\xowl{Thing}\! with \!\Res\! in our 
type system. This way we unify \!\Res\!, \!\xrdfs{Resource}\! and \!\xowl{Thing}\! as 
the top level of all resources. 
We do not consider any further features of OWL to be part of our type system.

\paragraph{SPARQL Protocol and RDF Query Language.}

The SPARQL suite of recommendations makes reference only to simple types. 
SPARQL Query~\cite{Harris2013} specifies the types of basic operations that 
are used to filter queries. For example, a regular expression can only apply 
to a string, and the sum to two integers is an integer. SPARQL Query 
treats all URIs as being of type URI. We also adopt this approach.

\paragraph{Open Graph Protocol.}

Facebook's Open Graph protocol uses an approach to Linked Data called 
microformats, where bits of RDF data are embedded into HTML documents. 
Microformats help machines to understand the content of the Web pages, which can 
be used to drive powerful searches. The Open Graph documentation states the 
following: ``properties have `types' which determine the format of their 
values."\footnote{\url{https://developers.facebook.com/docs/opengraph/property-types/} 
accessed on 27 March 2013.} In the terminology of the Open Graph 
documentation, the value of a property is the object of an RDF triple. The 
documentation explicitly includes the simple data types \!\xsd{string}\!, 
\!\xsd{integer}\!, etc, as types permitted to appear in the object position. 
This corresponds to the notion of type throughout this section.

\subsection{Local Type Checking for Dereferenced Linked Data}

We design our system such that, when a URI is dereferenced, only well typed 
queries are loaded into the triple store. This means that we can guarantee 
that all triples in the triple store are well typed.

Suppose that after dereferencing the URI \!\xres{Kazakhstan}\! we obtain the following two triples:
\[
\begin{array}{l}
\ntriple{\xres{Kazakhstan}}{\xdbp{demonym}}{\literal{Kazakhstani}{en}} 
\\
\ntriple{\xres{Kazakhstan}}{\xdbp{demonym}}{\xres{Kazakhstani}}
\end{array}
\]
Suppose also that the property \!\xdbp{demonym}\! has the type 
$\Property{\xsd{string}}$. The first triple is well typed, hence it is 
loaded into the store in the named graph \!\Kaz\!. However, the second triple 
is not well typed, hence would be ignored. No knowledge of other triples 
loaded into the store is required, i.e.\ our type system is local.

Since only well typed triples are loaded into the local triple store, scripts that use data in the local triple store can rely on type properties. Thus, for example, if a script consumes the object of any triple with \!{\xdbp{demonym}}\! as the property, then the script can assume that the term returned will be a string. This allows some static analysis to be performed by a type system for scripts. This observation is the basis of the type system in the next section.

\section{A Typed Scripting Language for Linked Data Consumers}
\label{section:system}

In this section, we define our language and type system. A grammar specifies the abstract syntax, and deductive rules specify the type system. We briefly outline the operational semantics, which is defined as a reduction system.

\subsection{Syntax}

We introduce a syntax for a typed 
high level scripting language that is used to consume Linked Data.

\paragraph{The Syntax of Types.}

A type is either a simple datatype, or it is a property that allows a simple datatype as its range, as follows:
\begin{gather*}
\simpleType \Coloneqq \Res \mid \xsd{string} \mid \xsd{decimal} \mid \xsd{dateTime} \mid \xsd{integer}
\\[8pt]
\ttype \Coloneqq \simpleType \mid \Property{\simpleType}
\qquad
\variable \Coloneqq \var{x} \mid \var{y} \mid \hdots
\qquad
\Gamma \Coloneqq \epsilon \mid \var{x} \colon \ttype, \Gamma
\end{gather*}

If a URI with a property type is used in the property position of a triple, then the object of that triple can only take the value indicated by the property type. Property types are assigned to URIs using the finite partial function $\Ont{\,\cdot\,}$ from URIs to property types. This partial function can be derived from ontologies or inferred from data.

Type environments are defined by lists of assignments of variables to types. 
As in popular scripting languages, such as Perl and PHP, variables begin 
with a dollar sign.

\paragraph{The Syntax of Terms and Expressions.}

Terms are used to construct RDF triples. Terms can be URIs, variables or literals of a simple datatype, each of which is drawn from a disjoint pool of lexemes.
\begin{gather*}
\uri \Coloneqq 
 \xres{URI} \mid
 \hdots 
\qquad
\textit{integer} \Coloneqq 99 \mid \hdots
\qquad
\decimal \Coloneqq 99.9 \mid 0.999e2 \hdots
\\[10pt]
\textit{string} \Coloneqq \mstring{WWV2013} \mid \literal{workshop}{en-gb} \hdots 
\qquad
\dateTime \Coloneqq \texttt{2013-06-6T13:00:00+01:00} \mid \hdots
\qquad
\\[10pt]
\langrange \Coloneqq \texttt{*} \mid \texttt{en} \mid \texttt{en-gb} \mid \hdots 
\qquad
\term \Coloneqq \variable \mid \uri \mid \lstring \mid \integer \mid \decimal \mid \dateTime
\\[10pt]
\textit{regex} \Coloneqq \texttt{WWV.*} \mid \hdots
\qquad
\expr \Coloneqq \term \mid \texttt{now} \mid \texttt{str}({\expr})  \mid \texttt{abs}(\expr) \mid {\expr + \expr} \mid {\expr - \expr}
\mid \hdots
\end{gather*}
Notice that strings may have a language tag, as defined by RFC4646~\cite{RFC4646}.
A language tag can be matched by a simple pattern, called a language range,
where \texttt{*} matches any language tag (e.g., \texttt{en} matches any 
dialect of English). Regular expressions over strings conform 
to the XPath recommendation~\cite{Malhotra2010}.

Expressions are formed by applying unary and binary functions over terms. 
The SPARQL Query recommendation~\cite{Harris2013} defines several standard 
functions including \texttt{str}, which maps any term to a string, and 
\texttt{abs}, which takes the absolute value of a number. The expression 
\texttt{now} represents the current date and time. The vocabulary of 
functions may be extended as required.

\paragraph{The Syntax of Scripts.}

Scripts are formed from boolean expressions, data and queries. They define a sequence of operations that use queries to determine what URIs to dereference.
\begin{gather*}
\begin{array}{rl}
\lboolean \Coloneqq & \lboolean \logor \lboolean \mid \lboolean \logand \lboolean \mid \neg \lboolean \\
               \mid & \regex{\expr}{\textit{regex}} \mid \langMatches{\expr}{\langrange} \mid \expr = \expr \mid \expr < \expr \mid \hdots 
\\[10pt]
\ltriples \Coloneqq & {\term}~{\term}~{\term} \mid \ltriples~\ltriples
\qquad \ldata \Coloneqq \texttt{graph}\,\term\,\{ \ltriples \} \mid \ldata~\ldata
\\[10pt]
\query \Coloneqq & \ldata \mid \lboolean \mid \query~\query \mid \query \union \query
\end{array}
\\[10pt]
\begin{array}{rlr}
\lscript \Coloneqq & \texttt{where}~\query~\lscript & \mbox{\it Satisfy a query pattern before continuing.} \\
              \mid & \fromnamed{\term}~\lscript & \mbox{\it Dereference a URI and load it into the local triple store.} \\
              \mid & \select{\variable \colon \ttype}~\lscript & \mbox{\it Select a binding for a variable to enable progress.} \\ 
              \mid & \iterate~\lscript & \mbox{\it Iteratively execute the script using separate data.} \\
              \mid & \cunit                          & \mbox{\it Successfully terminate.} \\
\end{array}
\end{gather*}

Data is represented as quadruples of terms, which always indicate the named graph. In SPARQL, the \texttt{graph} and \texttt{from named} keywords work in tandem. The \texttt{from named} keyword makes data from a named graph available, whilst keeping track of the context. The keyword \texttt{graph} allows the query to directly refer to the context. This contrasts to the \texttt{from} keyword in SPARQL which fetches data without keeping the context. We extend the meaning of the \texttt{from named} keyword so that the URI is dereferenced and makes the data available from that point onwards in the context of the URI that is referenced.

Boolean expressions are called filters in the terminology of SPARQL. We 
drop the keyword \texttt{filter}, because the syntax is unambiguous without 
it. Filters can be used to match a string with a regular expression or 
language tag, or compare two expressions of the same type.

Queries are constructed from filters and basic graph patterns indicated by 
the keyword \texttt{where}. The basic graph pattern should be matched using 
data in the local triple store. The basic graph pattern may contain 
variables. Variables in a script must be bound using the \texttt{select} 
keyword. In this scripting language, the \texttt{select} keyword acts like 
an existential quantifier. 
The variables bound by \texttt{select} are annotated with a type.
However these type annotations may be omitted by the programmer, since they can be algorithmically inferred using the type system. 

\subsection{The Type System}

\begin{figure}[b]
\xymatrix{
&
\Res
&
\\
\Property{\xsd{string}}\ar[ur]
&
\Property{\xsd{integer}}\ar[u]
&
\Property{\xsd{dateTime}}\ar[ul]
&
\Property{\Res}\ar[ull]
\\
&
\Property{\xsd{decimal}}\ar[u]
}
\caption{Subtype relations between types that can be assigned to URIs.}
\label{figure:subtype}
\end{figure}
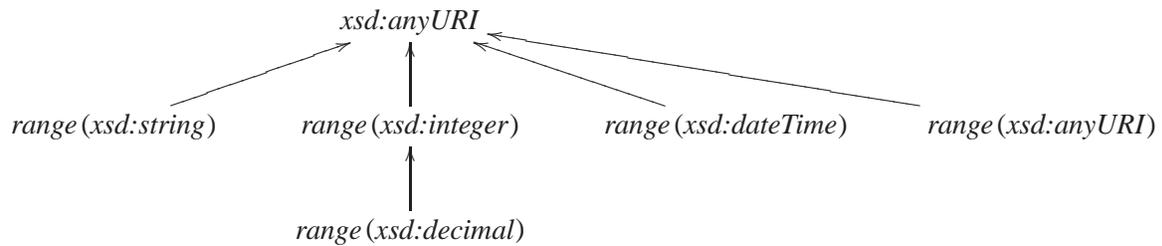

For existing implementations of SPARQL, a query that violates types 
silently fails, returning an empty result.
However, a type error is generally caused by an oversight by the programmer.
It would be helpful if a type 
error is provided at compiler time, indicating that a query has been designed such 
that the types guarantee that no result will ever be returned.
For this purpose, we introduce the type system presented in this section.

\paragraph{Subtypes.}

We define a subtype relation, that defines when one type can be treated as another type. The system indicates that \xsd{integer} is a subtype of \xsd{decimal}\!. It also defines property types to be contravariant, i.e.\ they reverse the direction of subtyping. In particular, if a property permits decimal numbers in the object position, then it also permits integers in the object position. 
\begin{gather*}
\begin{prooftree}
\justifies
\vdash \xsd{integer} \leq \xsd{decimal}
\end{prooftree}
\qquad
\begin{prooftree}
\justifies
\vdash \mtype \leq \mtype
\end{prooftree}
\\[12pt]
\begin{prooftree}
\vdash \simpleType_1 \leq \simpleType_2
\justifies
\vdash \Property{\simpleType_2} \leq \Property{\simpleType_1}
\end{prooftree}
\qquad
\begin{prooftree}
\justifies
\vdash \Property{\simpleType} \leq \Res
\end{prooftree}
\end{gather*}
The subtype relations between types that can be assigned to URIs are 
summarised in Figure~\ref{figure:subtype}.

\paragraph{The type system for terms and expressions.}
Types for terms assign types to lexical tokens and assign types to properties using the partial function $\Ont{\,\cdot\,}$ from URIs to types. Types for expressions ensure that operations are only applied to resources of the correct type.
\begin{gather*}
\begin{prooftree}
\vdash \ttype_0 \leq \ttype_1
\justifies
\Gamma, \var{x} \colon \ttype_0 \vdash \var{x} \colon \ttype_1
\end{prooftree}
\qquad
\begin{prooftree}
\vdash \Ont{\uri} \leq \ttype
\justifies
\Gamma \vdash \uri \colon \ttype
\end{prooftree}
\qquad
\begin{prooftree}
\vdash \xsd{integer} \leq \simpleType
\justifies
\Gamma \vdash \integer \colon \simpleType
\end{prooftree}
\\[12pt]
\begin{prooftree}
\justifies
\Gamma \vdash \decimal \colon \xsd{decimal}
\end{prooftree}
\qquad
\begin{prooftree}
\justifies
\Gamma \vdash \lstring \colon \xsd{string}
\end{prooftree}
\qquad
\begin{prooftree}
\justifies
\Gamma \vdash \dateTime \colon \xsd{dateTime}
\end{prooftree}
\\[12pt]
\begin{prooftree}
\justifies
\Gamma \vdash \texttt{now} \colon \xsd{dateTime}
\end{prooftree}
\qquad
\begin{prooftree}
\Gamma \vdash \expr_1 \colon \simpleType
\quad
\Gamma \vdash \expr_2 \colon \simpleType
\quad
\vdash \simpleType \leq \xsd{decimal}
\justifies
\Gamma \vdash \expr_1 + \expr_2 \colon \simpleType
\end{prooftree}
\\[12pt]
\begin{prooftree}
\Gamma \vdash \expr \colon \simpleType
\justifies
\Gamma \vdash \texttt{str}(\expr) \colon \xsd{string}
\end{prooftree}
\qquad
\begin{prooftree}
\Gamma \vdash \expr \colon \simpleType
\quad
\vdash \simpleType \leq \xsd{decimal}
\justifies
\Gamma \vdash \texttt{abs}(\expr) \colon \simpleType
\end{prooftree}
\end{gather*}
The above types can easily be extended to cover all functions in the SPARQL recommendation, such as \texttt{seconds} which maps an \xsd{dateTime} to a \xsd{decimal}\!. Our examples include the custom function \texttt{haversine} which maps four expressions of type \xsd{decimal} to one \xsd{decimal}\!.

The type system for filters follows a similar pattern to expressions. For example, the type system ensures that only terms of the same type can be compared.
\begin{gather*}
\begin{prooftree}
\Gamma \vdash \expr \colon \xsd{string}
\justifies
\Gamma \vdash \regex{\expr}{\textit{regex}}
\end{prooftree}
\qquad
\begin{prooftree}
\Gamma \vdash \expr \colon \xsd{string}
\justifies
\Gamma \vdash \langMatches{\expr}{\langrange}
\end{prooftree}
\\[12pt]
\begin{prooftree}
\Gamma \vdash \expr_1 \colon \simpleType
\quad
\Gamma \vdash \expr_2 \colon \simpleType
\justifies
\Gamma \vdash \expr_1 = \expr_2
\end{prooftree}
\qquad
\begin{prooftree}
\Gamma \vdash \expr_1 \colon \simpleType
\quad
\Gamma \vdash \expr_2 \colon \simpleType
\justifies
\Gamma \vdash \expr_1 < \expr_2
\end{prooftree}
\\[12pt]
\begin{prooftree}
\Gamma \vdash \lboolean_0 
\quad
\Gamma \vdash \lboolean_1 
\justifies
\Gamma \vdash \lboolean_0 \logand \lboolean_1 
\end{prooftree}
\qquad
\begin{prooftree}
\Gamma \vdash \lboolean_0 
\quad
\Gamma \vdash \lboolean_1 
\justifies
\Gamma \vdash \lboolean_0 \logor \lboolean_1 
\end{prooftree}
\qquad
\begin{prooftree}
\Gamma \vdash \lboolean 
\justifies
\Gamma \vdash \lognot \lboolean 
\end{prooftree}
\end{gather*}

\paragraph{The type system for data.}

The types for terms are used to restrict the subject of triples to URIs, and the object to the type prescribed by the property. For quadruples, the named graph is always a URI, as prescribed by the type system.
\begin{gather*}
\begin{prooftree}
\Gamma \vdash \term_1 \colon \Res
\quad
\Gamma \vdash \term_2 \colon \Property{\simpleType}
\quad
\Gamma \vdash \term_3 \colon \simpleType
\justifies
\Gamma \vdash {\term_1}~~{\term_2}~~{\term_3}
\end{prooftree}
\\[12pt]
\begin{prooftree}
\Gamma \vdash \textit{triples}_1
\quad
\Gamma \vdash \textit{triples}_2
\justifies
\Gamma \vdash \textit{triples}_1~\textit{triples}_2
\end{prooftree}
\qquad
\begin{prooftree}
\Gamma \vdash \term_0 \colon \Res
\quad
\Gamma \vdash \textit{triples}
\justifies
\Gamma \vdash \texttt{graph}\,\term_0\,\{\,\textit{triples}\,\}
\end{prooftree}
\\[12pt]
\begin{prooftree}
\Gamma \vdash \term \colon \Property{\simpleType}
\justifies
\Gamma \vdash \ntriple{\term}{\xrdfs{range}}{\simpleType}
\end{prooftree}
\qquad
\begin{prooftree}
\Gamma \vdash \term \colon \Property{\Res}
\justifies
\Gamma \vdash \ntriple{\term}{\type}{\xowl{ObjectProperty}}
\end{prooftree}
\qquad
\end{gather*}

We include special type rules for two particular forms of triples. The first, from the RDF Schema vocabulary~\cite{Brickley2004}, allows the range of a property to be explicitly prescribed as a datatype. The second, from the OWL vocabulary~\cite{Hitzler2012}, prescribes when the range of a property is a URI. These two rules are useful for future work in type inference. In particular, when type information is not available, these rules would help infer the minimum partial function $\Ont{\,\cdot\,}$ that allows a dataset to be typed. 

\paragraph{The Type System for scripts.}
Scripts can be type checked using our type system. The rule for the 
\texttt{from named} keyword checks that the term to dereference is a URI. 
The rule for the \texttt{select} makes use of an assumption in the type 
environment to ensure that the variable is used consistently in the rest 
of the script, where the variable is bound.
\begin{gather*}
\begin{prooftree}
\Gamma \vdash \query_1 
\quad
\Gamma \vdash \query_2 
\justifies
\Gamma \vdash \query_1~\query_2 
\end{prooftree}
\qquad
\begin{prooftree}
\Gamma \vdash \query_1 
\quad
\Gamma \vdash \query_2 
\justifies
\Gamma \vdash \query_1 \union \query_2 
\end{prooftree}
\qquad
\begin{prooftree}
\Gamma \vdash \query
\quad
\Gamma \vdash \lscript
\justifies
\Gamma \vdash \texttt{where}~\query~\lscript
\end{prooftree}
\\[12pt]
\qquad
\begin{prooftree}
\Gamma \vdash \term \colon \Res
\quad
\Gamma \vdash \lscript
\justifies
\Gamma \vdash \fromnamed{\term}~\lscript
\end{prooftree}
\qquad
\begin{prooftree}
\Gamma, \var{x} \colon \ttype \vdash \lscript
\justifies
\Gamma \vdash \select{\var{x} \colon \ttype}~\lscript
\end{prooftree}
\qquad
\begin{prooftree}
\Gamma \vdash \lscript
\justifies
\Gamma \vdash \iterate~\lscript
\end{prooftree}
\qquad
\Gamma \vdash \cunit
\end{gather*}

If a script is well typed with an empty type environment, then 
all the variables must be bound using the rule for the \texttt{select} quantifier.
Furthermore, since the \texttt{select} quantifiers does not appear in data, no variables can appear in well typed data.
We assume that scripts and data are executable only if they are
well typed with respect to an empty type environment.

\subsection{Examples of Well Typed Scripts}
We consider some well typed scripts, and suggest some errors that our type system avoids.

The following script is well typed. The script finds resources in any named graph that have a label in the Russian language. It then dereferences the resources. The script is iterated as many times as the implementation feels necessary, without revisiting data.
\[
\script{
 \iterate 
 ~~ \script{
      \select{\var{g} \colon \Res, \var{x} \colon \Res, \var{y} \colon \xsd{string}} \\
      \where{
       \nquad{\var{g}}{\var{x}}{\xrdfs{label}}{\var{y}} \\
       \langMatches{\var{y}}{\texttt{ru}}
      } \\
      \fromnamed{\var{x}}
    }
}
\]
Note that if we use the variable $\var{y}$ instead of $\var{x}$ in the 
\texttt{from named} clause, then the script could not be typed. The 
variable $\var{y}$ would need to be both a string and a URI, but it can only be one or the other. 
We assume that $\Ont{\xrdfs{label}} = 
\Property{\xsd{string}}$.

The following well typed script looks in two named graphs. In named graph \!\Kaz\!, it looks for properties with \!\Kaz\! as the object, and in the named graph \!\xdbp{}\! it looks for properties that have either a label or comment that contains the string \texttt{"location"}.
\[
\script{
 \select{\var{p} \colon \Property{\Res}, \var{y} \colon \xsd{string}, \var{z} \colon \xsd{string}} \\
 \where{
    \{\\~~
     \begin{array}{l}
     \nquad{\xdbp{}}{\var{p}}{\xrdfs{label}}{\var{y}} \\
     \union \\
     \nquad{\xdbp{}}{\var{p}}{\xrdfs{comment}}{\var{y}}
     \end{array}\\
    \}\\
    \nquad{\Kaz}{\var{z}}{\var{p}}{\Kaz} \\
    \regex{\var{y}}{\texttt{location}} \logand
    \langMatches{\var{y}}{\texttt{en}}
 } \\
 \fromnamed{\var{p}}
}
\]
We must assume that $\Ont{\xrdfs{comment}} = \Property{\xsd{string}}$. If we assume otherwise, then the above query is not well typed. Note also that property $\var{p}$ must be of type $\Property{\Res}$.
We cannot assign it a more general type such as \Res, although we can use it as a resource. 
\begin{figure}[h]
\[
\script{
 \fromnamed{\xres{Almaty}} \\
 \select{\var{almalat} \colon \xsd{decimal}, \var{almalong} \colon \xsd{decimal}} \\
 \where{
  \nquad{\xres{Almaty}}{\xres{Almaty}}{\xgeo{lat}}{\var{almalat}} \\
  \nquad{\xres{Almaty}}{\xres{Almaty}}{\xgeo{long}}{\var{almalong}}
 } \\
 \fromnamed{\xres{Kazakhstan}} \\
 \iterate 
 ~ \script{
    \select{\var{loc} \colon \Res} \\
    \where{
      \nquad{\Kaz}{\var{loc}}{\xdbp{location}}{\xres{Kazakhstan}}
    } \\
    \fromnamed{\var{loc}} \\
    \select{\var{lat} \colon \xsd{decimal}, \var{long} \colon \xsd{decimal}} \\
    \where{
      \nquad{\var{loc}}{\var{loc}}{\xgeo{lat}}{\var{lat}} \\
      \nquad{\var{loc}}{\var{loc}}{\xgeo{long}}{\var{long}} \\
      \texttt{haversine}(\var{lat},\var{long},\var{almalat},\var{almalong}) < 100
    } \\
    \iterate 
    ~  \script{
         \select{\var{person} \colon \Res} \\
         \where{
           \nquad{\var{loc}}{\var{person}}{\xdbp{birthPlace}}{\var{loc}}
         } \\
         \fromnamed{\var{person}}
       }
}
}
\]
\caption{Get data about people born in places in Kazakhstan less than 100km from Almaty.}
\label{figure:example}
\end{figure}

Finally, consider the substantial example of Figure~\ref{figure:example}. We assume that the function \texttt{haversine} calculates the distance (in km) between two points identified by their latitude and longitude. The script pulls in data about places located in Kazakhstan. It then uses this data to pull in more data about people born in places less than 100km from Almaty.


\subsection{Operational Semantics for the System}

In related work, we extensively study the operational semantics of 
languages for Linked Data that are related to the language proposed in this 
work~\cite{Ciobanu2012,Dezani2012,Horne2011,Sassone2013,Gibbins2011}. In the remaining 
space, we briefly sketch the operational semantics for our scripting 
language.


Systems are data and scripts composed in parallel by using the $\|$ operator. 
The main rules are the rules for dereferencing URIs, for selecting bindings, and for interactions between a query and data.
The rules can be applied in any context.
\begin{gather*}
\begin{prooftree}
\vdash \texttt{graph}\,\uri\,\{ \ltriples \}
\justifies
\fromnamed{\uri}~\lscript
\longrightarrow
\lscript \cpar \texttt{graph}\,\uri\,\{ \ltriples \}
\end{prooftree}
\\[12pt]
\begin{prooftree}
\vdash \term \colon \ttype
\justifies
\select{\var{x} \colon \ttype}~\lscript
\longrightarrow
\lscript\sub{\var{x}}{\term}
\end{prooftree}
\qquad
\begin{prooftree}
\ldata \leq \query
\justifies
\texttt{where}\,{\query}~\lscript
\cpar
\ldata
\longrightarrow
\lscript
\cpar
\ldata
\end{prooftree}
\end{gather*}
The rules for dereferencing data involves a dynamic type check of 
arbitrary data that arrives as a result of dereferencing a URI. The 
\texttt{select} rule also performs a dynamic check to ensure that the term 
substituted for a variable is of the correct type. The query rule reads 
data that matches the query pattern (using a preorder $\leq$ over queries), 
without removing the data.

\medskip
The following result proves that a well typed term will always reduce to a well typed term. Thus the type system is sound with respect to the operational semantics.
\begin{theorem}
If $\vdash \lsystem_1$ and $\lsystem_1 \longrightarrow \lsystem_2$, 
then $\vdash \lsystem_2$.
\end{theorem}

\begin{proof}
We provide a proof sketch covering the cases for only the three rules above.

Consider the operational rule for \texttt{from named}. Assume that both $\vdash \fromnamed\,\uri~\lscript$ and $\vdash \texttt{graph}\,\uri\,\{ \ltriples \}$ hold. By the type rule for \texttt{from named}, $\vdash \lscript$ must hold. Thus, by the type rule for parallel composition, $\vdash \lscript \cpar \texttt{graph}\,\uri\,\{ \ltriples \}$.

\textit{Lemma.} If $\vdash \term \colon \ttype$ and $\var{x} \colon \ttype \vdash \lscript$, then $\vdash \lscript\sub{\var{x}}{\term}$, by structural induction. \hfill\qed

Consider the operational rule for \texttt{select}. Assume that $\vdash \term \colon \ttype$ and $\vdash \select{\var{x} \colon \ttype}~\lscript$ hold. By the type rule for \texttt{select}, $\var{x} \colon \ttype \vdash \lscript$ must hold. Hence, by the above lemma, $\vdash \lscript\sub{\var{x}}{\term}$ holds.

Consider the operational rule for \texttt{where}. Assume that $\vdash \texttt{where}\,\query~\lscript \cpar \ldata$ holds. By the type rule for parallel composition, $\vdash \texttt{where}\,\query~\lscript$ and $\vdash \ldata$ must hold, and, by the type rule for \texttt{where}, $\vdash query$ and $\vdash \lscript$ must hold. Hence $\vdash \lscript \cpar \ldata$ holds.
\end{proof}

Further cases and results will be covered in an extended paper.
Future work includes implementing an interpreter for the language based on the operational semantics, and developing a minimal type inference algorithm~\cite{Schwartzbach1991} based on the type system.

\section{Conclusion}

As the Web of Linked Data grows, the state of the art 
for commercial Linked Data solutions is also advancing.
State of the art of triple stores allow efficient execution of 
queries at scale. Furthermore, the back end of commercial solutions such as 
Virtuoso~\cite{Erling2007} and the Information Workbench~\cite{Haase2012} 
can extract data from diverse sources. This allows us to take a liberal 
view of the Web of Data, where data is drawn from data APIs provided by popular services 
from Twitter, Facebook and Google. Through experience with master 
students, we found that developers with experience of a Web 
development platform such as .NET or Ruby-on-Rails can assemble a 
front end in a matter of days, simply by shifting their query language 
from SQL to SPARQL.

Linked Data enables processes to
programmatically crawl the Web of Linked Data, pulling data from diverse 
sources and removing boundaries between datasets. The data pulled from the Web forms a local view of the Web of 
Linked Data that is tailored to a particular application. We found that 
existing programming environments, consisting of a general purpose language 
and a library, obstructed swift development of such processes with many 
low level details. This exposes the need for a high level language that 
makes scripting background processes that consume Linked Data easy. 

In this work, we introduce a domain specific high level language for 
consuming Linked Data. Domain specific languages are designed at a level 
of abstraction that simplifies programming tasks in the domain. In our 
domain specific language, key operations such as queries are primitive, 
meaning that basic syntactic checks can be performed. It is also easier to 
perform static analysis over a domain specific language.
We take care to design the syntax of the scripting language such that it resembles the SPARQL 
recommendations~\cite{Harris2013}, to appeal to the target Web developers.

We introduce a simple but effective type system for our language. The type 
system is based on the fragment of the SPARQL, RDF Schema and OWL 
recommendations that deals with simple data types. The applications can 
statically identify simple errors such as attempts to dereference a 
number, or attempts to match the language tag of a URI. For static type 
checking of scripts, the data loaded into the system must be dynamically 
type checked to ensure that properties have the correct literal value or a 
URI as the object.
The dynamic type checks do not impose significant restrictions on the data consumed.
Most datasets, including data from DBpedia, conform to 
this typing pattern. Further weight is added by the Facebook Open Graph 
protocol, which demands typing at exactly the level we deliver.

\paragraph{Acknowledgements.}
The work was supported by a grant of the Romanian National Authority for Scientific Research, CNCS-UEFISCDI, project number PN-II-ID-PCE-2011-3-0919.

\bibliographystyle{eptcs}
\bibliography{bibliography}
\end{document}